\documentclass[12pt]{article}
\usepackage{amsmath}
\usepackage{amssymb}
\usepackage{graphicx}
\usepackage{tikz}
\usetikzlibrary{arrows}
\usepackage{verbatim}
\usepackage{natbib}

\newcommand{\Var}{\text{Var}}

\newcommand{\Exp}{\text{Exp}}
\newcommand{\ExpP}{\Exp_P}
\newcommand{\ExpBone}{\Exp_{B1}}
\newcommand{\ExpBtwo}{\Exp_{B}}
\newcommand{\VarP}{\Var_P}

\newcommand{\VarBtwo}{\Var_{B}}

\newcommand{\GraphG}{{\mathbb G}}  
\newcommand{\originalgraphW}{{\mathbb W}}
\newcommand{\adjacentTedges}{ \#{\mathbb T}}
\newcommand{\adjacentCedges}{ \#{\mathbb C}}
\renewcommand{\vec}[1]{\mathbf{#1}}
\newcommand{\vectorC}{\vec{C}}
\newcommand{\vectorT}{\vec{T}}
\newcommand{\vectorTC}{\vec{TC}}
\newcommand{\indegree}{\text{ID}}
\newcommand{\outdegree}{\text{OD}}
\newcommand{\indegreev}{\indegree(v)}
\newcommand{\outdegreev}{\outdegree(v)}
\newcommand{\Case}[1]{\text{*Case #1*}}
\newcommand{\Eterm}[1]{\text{e#1}}
\newcommand{\remterm}[1]{\text{reme#1}}
\usepackage{amsthm}
\newtheorem{theorem}{Theorem}

\begin{document}
\title{Exact Bootstrap and Permutation Distribution of Wins and Losses in a Hierarchical Trial\footnote{This is a revision of an earlier preprint dated Januay 28, 2019.  The differences are mainly in organization; the results remain the same.}}
\author{William N. Anderson\footnote{Carpinteria, California, USA. Email: WNilesAnderson@gmail.com. LinkedIn: william-anderson-46384b7} and Johan Verbeeck\footnote{I-Biostat, University Hasselt, Belgium.  Email: johan.verbeeck@uhasselt.be}}
\date{November 24, 2019}
\maketitle
\begin{abstract}
Finkelstein-Schoenfeld, Buyse, Pocock,  and other authors have developed generalizations of the Mann-Whitney test that allow for pairwise patient comparisons to include a hierarchy of measurements. Various authors present either asymptotic or randomized methods for analyzing the wins.  We use graph theory concepts to derive exact means and variances for the number of wins, as a replacement for approximate values obtained from bootstrap analysis or random sampling from the permutation distribution.  The time complexity of our algorithm is $O(N^2)$, where $N$ is the total number of patients.  In any situation where the mean and variance of a bootstrap sample are used to draw conclusions, our methodology will be faster and more accurate than the randomized bootstrap or permutation test. 
\end{abstract}


2010 MATHEMATICS SUBJECT CLASSIFICATIONS: 62G09, 62N99

\section{Background}
\par The context is the hierarchical clinical trial methodology of Finkelstein and Schoenfeld  \citeyearpar{FinkelsteinSchoenfeldCombining}, which is a far-reaching generalization of the classical Gehan-Wilcoxon test~\citep{Gehancensored}. The Finkelstein-Schoenfeld test produces a {\em p}-value, based on the the permutation distribution of the trial arms.   The {\em p}-value computation uses the difference between treatment wins and control wins, and does not consider the two items separately.  

Other authors have considered similar analyses. References include Buyse \citeyearpar{BuyseGeneralized}, Dong \citeyearpar{DongWinRatioVariance}, Luo et al. \citeyearpar{LuoAlternative, LuoWeighted}, Bebu et al. \citeyearpar{BebuUstatistics}, and Pocock et al. \citeyearpar{PocockWinRatio}. In these papers the variances for wins and their differences are computed using either asymptotic formulas or randomized sampling. Inferences are then drawn from the computed means and variances. 

We use graph theory concepts to derive the exact mean and variance of the trial arm wins for both bootstrap and permutation distribution analysis. The algorithm complexity is $O(N^2)$ in both time and space. 
This methodology eliminates the need for asymptotic evaluation of variance matrices. Our algorithm will be both faster and more accurate than randomized bootstrap or permutation tests in any situation where inferences are drawn from the mean and variance of the trial arm wins. For ratios the situation is slightly more complicated. We present R$^\circledR$ and SAS$^\circledR$ code to implement the computations. 

In addition to the theoretical treatment, we illustrate with a simple example.  R and SAS code, and a more detailed treatment of the example, are contained in appendices.

\subsection{A bit of notation} We will present expected values and variances under two different methodologies: the permutation distribution of trial arms and the bootstrap distribution with sampling from the trial arms separately. We will use the notations $\ExpP$, $\VarP$,  $\ExpBtwo$, and $\VarBtwo$ to clarify which methodology is involved in the specific formulas.

\section {The Finkelstein-Schoenfeld Test} \label{FinkelsteinSchoenfeldsectionmainline}

The  Finkelstein-Schoenfeld test is based on a hierarchical composition of a number of measures. For each measure every patient is compared to every other patient in a pairwise manner.  We define a score, $ u_{ij}$, which is chosen to reflect whether patient $i$ has had the more favorable outcome than patient $j$. The concept of more favorable can involve censoring, missing data,  and can include a threshold; accordingly patients who cannot be compared may not necessarily be tied.  The score is 
\begin{itemize}
\item[-1] if patient $i$ has a less favorable outcome than patient $j$;
\item[0] if the patients cannot be compared, or are tied;
\item[+1] if patient $i$ has a more favorable outcome than patient $j$. 
\end{itemize}

The scores are computed separately for each measure in the hierarchy. The overall comparison  for the pair is the first non-zero comparison, where first is defined by the prespecified hierarchy. It is quite possible that for many pairs of patients the overall comparison will remain $0$. It is also possible to have non-transitive comparisons; that is, three patients such that $ u_{ij} = 1$, $ u_{jk} = 1$, and $ u_{ki} = 1$ \citep{VerbeeckGeneralized}.

The  Finkelstein-Schoenfeld test is a score test based on the sum of the scores for the treated group. Suppose there are $N$ subjects in the trial, with $m$ subjects in the treatment group and $n$ subjects in the control group. Let $D_i =1$ for subjects in the treatment group, and $D_i =0$ for patients in the control group. Using the $ u_{ij}$ for every pair of patients defined above, we assign a score to each subject, $ U_i = \sum_j u_{ij}$.

The test is now based upon
\begin{equation} \label{FSdefinitionmainline}
 FS = \sum_{i = 1}^N D_i U_i  .
\end{equation}

It it easy to see that $\ExpP(FS) = 0$, and it can be shown that 
\begin{equation} \label{FSVarmainline} 
\VarP(FS) = \dfrac{mn}{N(N-1)}\sum_{i = 1}^N U_i^2.
\end{equation}

The expected value and variance above  furnish the needed information to compute a $z$-statistic and hence a $p$-value, under the null hypothesis of no difference between the trial arms.  The formulas are derived in  \cite{FinkelsteinSchoenfeldCombining}.

The Partner IB and TAVR UNLOAD trials have used this methodology prospectively \cite{PartnerIB1year},  \cite{UNLOADDesign}. The method has also been used to furnish alternative analyses to trials with other prespecified analyses \cite{BuyseGeneralized}, \cite{DongWinRatioVariance}, \cite{FinkelsteinSchoenfeldCombining}, \cite{PocockWinRatio}.

\section{Wins and Losses}\label{winsandlossessection}

The outcome of the trial described in section \ref{FinkelsteinSchoenfeldsectionmainline} can be described by two items:

\begin{itemize}
\item An $ N \times N$ outcome matrix $U$, which is skew and has entries in the set $\{-1, 0, +1\}$.
  
\item A trial arm vector $D$, which has the value 1 for patients in the treatment arm and 0 for patients in the control arm. There will be $m$ patients in the treatment arm, and $n$ patients in the control arm. 
\end{itemize}

One method to produce such an outcome matrix is analyzing a single variable by use of the Mann-Whitney test, or the Gehan-Wilcoxon test, which allows for censoring. The Finkelstein-Schoenfeld hierarchical analysis allows for consideration of more variables.   In the remainder of this manuscript we consider that the outcome matrix $U$ has been produced by some unspecified mechanism; the only restriction is the skewness mentioned above.  

We will analyze the means and variances of wins using two different models. In both cases the analysis is based on the observed $U$ matrix and the trial arm assignments. 

\begin{itemize}

\item {\em Permutation distribution:}  In this analysis all possible permutations of $m$ treatment assignments and $n$ control assignments are considered, and the actual assignments play no special role.  This permutation analysis is the analysis used in the classical Wilcoxon test \citep{HollanderWolfenonparametric}, and its generalization as the Gehan-Wilcoxon test \citep{Gehancensored} used to analyze a single censored variable. It is also the methodology that was used above in the Finkelstein-Schoenfeld analysis \citeyearpar{FinkelsteinSchoenfeldCombining}.  

\item{\em Bootstrap:} In this analysis a new dataset is obtained by random sampling, with repetition, separately from each of the observed treatment groups. For each sample the numbers of treatment and control wins are computed, and the mean and variance of these numbers over many samples are computed. These means and variances are used to draw inferences about the true numbers of wins. There are $m^m n^n$ possible samples, each equally likely.

\item As an alternative to actually performing the bootstrap sampling, we compute the mean and variance over all possible bootstrap samples, using an $O(N^2)$ algorithm to be described below. This method will be more accurate than actually performing the bootstrap, because the randomization error is eliminated. It will also be faster than evaluating a large number of bootstrap samples, because evaluating only one bootstrap sample is already $O(N^2)$. 

\end{itemize}

\subsection{The graphical model}

For narrative purposes it will prove convenient to think of the outcome matrix $U$ as the adjacency matrix of a directed graph ${\mathbb G}$. Our terminology is reasonably standard, and follows Wikipedia \citeyearpar{wikigraphdefinitions}.

\begin{itemize}
\item If $u_{ij} = 1$, the edge $e = (i, j)$ is an ordered pair of distinct vertices. The {\em head} of $e$ is  vertex $j$ and the {\em tail} is vertex $i$. If $u_{ij} \ne 1$, there is no corresponding edge.  In a pictorial representation, we draw an arrow from vertex $i$ to vertex $j$. If vertex $v$ is the head or tail of edge $e$, then we say that $e$ is {\em adjacent} to $v$.  Let $E$ denote the total number of edges of ${\mathbb G}$.  

\item  For a vertex $v$, the {\em indegree} $\indegreev$ is the number of edges whose head is the vertex $v$; equivalently the number of $1$ entries in the associated column of $U$, or the number of pair comparisons that represent a loss for that patient. Similarly the {\em outdegree} $\outdegreev$ is the number of edges whose tail is the vertex $v$; equivalently the number of $1$ entries in the associated row of $U$, or the number of pair comparisons that represent a win for that patient.

\item For a specific trial arm assignment $D$, trial arm wins are defined in the following manner.

\begin{itemize}

\item An edge $e = (i, j)$  is a {\em win} for a treatment patient if  $D_i = 1$ and $D_j = 0$. 
\begin{itemize}
\item Define the $E \times 1$ vector $\vectorT$ by $T_e = k$ if edge $e$ corresponds to $k$ wins for treatment. 
  
\item Let $W_T$ = $\sum_e T_e$ be the total number of wins for treatment patients.
\end{itemize}

\item An edge $e = (i, j)$  is a {\em win} for a control patient if  $D_i = 0$ and $D_j = 1$. The vector $\vectorC$ is defined analogously to the vector $\vectorT$.

\item We use the notations ${\originalgraphW}_T$ and ${\originalgraphW}_C$ for the numbers of treatment and control wins in the observed data, as opposed to the numbers in a particular permutation or bootstrap repetition.

\item The vector $\vectorTC$ is the $2E \times 1$ vector with $\left[TC_1, \dots, TC_E\right] = \newline \left[T_1, \dots, T_E\right]$, and $\left[TC_{E + 1}, \dots, TC_{2E}\right] = \left[C_1, \dots, C_E\right]$. 

\item In analyzing the observed data and in the permutation distribution analysis, the vector $\vectorTC$ will have entries in the set $\{0, 1\}$.  Because bootstrap sampling is with repetition, values $> 1$ are possible in the bootstrap analysis.

\end{itemize}
\end{itemize}

Figure \ref{smallexamplecaption} gives a small example.   The U-matrix is

\begin{equation} \label{smallexampleUmatrix}
U = \begin{bmatrix}  
     \hphantom{-}0    &                  -1    &\hphantom{-}0   &\hphantom{-}0    &\hphantom{-}1  \\ 
     \hphantom{-}1     &\hphantom{-}0   &\hphantom{-}1   &\hphantom{-}0     &                  -1 \\   
     \hphantom{-}0     &                   -1   &\hphantom{-}0   &\hphantom{-}1     &\hphantom{-}0  \\
     \hphantom{-}0     &\hphantom{-}0  &                    -1    &\hphantom{-}0    &                  -1  \\ 
                       -1    &\hphantom{-}1     &\hphantom{-}0    &\hphantom{-}1    &\hphantom{-}0  
\end{bmatrix}.
\end{equation}

 This small example illustrates  a $U$ matrix that cannot be obtained from the Mann-Whitney or Gehan-Wilcoxon analysis. First, notice that patient 1 has a more favorable outcome than patient 5, and patient 5 has a more favorable outcome than patient 4; however patient 1 does not have a more favorable outcome than patient 4.  Second, note that patients 1, 5, and 2 form a directed cycle, so that there is no way to determine which is superior.  It is not difficult to create a hierarchy that will yield these results; we leave the details as an exercise for the reader. 

\begin{figure}[h!]
    \centering
\begin{tikzpicture}[->,>=stealth',shorten >=1pt,auto,node distance=3cm, thick,main node/.style={circle,draw,font=\sffamily\Large\bfseries, fill=orange}]
  \node[main node] (1) {1};
  \node[main node] (2) [right of=1] {2};
  \node[main node] (3) [right of=2] {3};
  \node[main node] (5) [below of=2] {5};
  \node[main node] (4) [below of=3] {4};
 \path[every node/.style={font=\sffamily\small}]
    (1) edge node {} (5)
    (2) edge node {} (1)
         edge node {} (3)
    (3) edge node {} (4)
    (5) edge node {} (4)
        edge node {} (2); 
\end{tikzpicture}
 \caption{Small Graph example}\label{smallexamplecaption}
\end{figure}
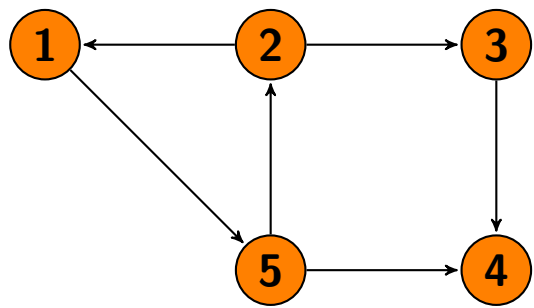

\clearpage

\subsection{Analysis using the graphical model}

For further analysis, we consider the outcome matrix $U$ to be fixed, and examine the behavior of the wins $W_T$ and $W_C$ under both analysis models. The  goal of this analysis is to compute the means and variances of $W_T$ and $W_C$, and their covariance.  

The analysis follows the same general pattern for both methods.
\begin{itemize}

\item The expected values $\Exp(\vectorTC)$ and $\Exp(\vectorTC(\vectorTC)^t)$ will be determined, from which we compute $\Var(\vectorTC) = \Exp(\vectorTC(\vectorTC)^t) - \Exp(\vectorTC)(\Exp(\vectorTC))^t$.  

\item For each pair of edges $(e, f)$, not necessarily distinct, the expected values $\Exp(T_eT_f)$, $\Exp(T_eC_f)$,  $\Exp(C_eT_f)$, and $\Exp(C_eC_f)$ can be computed using counting arguments.  The formulas will be different for different edge pairs, depending on the trial arm assignments and the geometric relationship of the edges.  

We note that for edge pairs in certain configurations, it will happen that $\Exp(T_eC_f) \ne \Exp(C_eT_f)$. This inequality is not a violation of the symmetry of the variance matrix; the latter requires that $\Exp(T_eC_f) = \Exp(C_fT_e)$ for all edge pairs $(e, f)$.

\item These computations are sufficient to compute all the terms of the expected value and variance of the vector $\vectorTC$. Define $S_{TT}, S_{TC}, S_{CT},$ and  $S_{CC}$ by
\begin{equation*} 
\Var(\vectorTC) = \begin{bmatrix} S_{TT} & S_{TC} \\ S_{CT} & S_{CC}\end{bmatrix},
\end{equation*}
where the partitioning corresponds to $\vectorT$ and $\vectorC$.  Finally
\begin{equation} \label{finalvariance}
\Var\left(\begin{bmatrix}W_T \\ W_C\end{bmatrix}\right) = \begin{bmatrix} \sum S_{TT} & \sum S_{TC} \\  \sum S_{CT} & \sum S_{CC} \end{bmatrix}.
\end{equation}
\begin{itemize} 
\item Because the final variance matrix is symmetric, we do not explicitly need the individual terms $\Exp(C_eT_f)$.
\item Since the number of potential edges is $O(N^2)$, the number of edge pairs is $O(N^4)$, and it would not be practical to explicitly compute $\Var(\vectorTC)$. Instead we count the number of times each geometric configuration of edges and trial arm assignments occurs, and compute \eqref{finalvariance} directly.  The resulting algorithm will be  $O(N^2)$, as will be seen below.

\end{itemize}

\item In all cases, the calculation of the expected value and variance terms is an elementary calculation involving binomial coefficients. In the text below we spare the reader the details; gentler derivations are available from the authors.  

\end{itemize}

\section{Wins analysis using the permutation distribution of trial arms} \label{AnalysisPermutation}

\begin{theorem} \label{permutationtheorem} The mean and variance of wins under the permutation distribution of trial arms can be computed in $O(N^2)$ time.
\end{theorem}
\begin{proof} \

\subsection{Expectations}

\begin{itemize}
\item For a single edge $e$, $\ExpP(T_e) = \ExpP(T_e) =  mn/N(N-1)$.
\item  For the entire graph $\GraphG$, $\ExpP(W_T) =  \ExpP(W_C) = Emn/N(N-1)$.
\item It follows that $\ExpP(W_T - W_C) = 0$, which is a restatement of $\ExpP(FS) = 0$, as was mentioned above.
\end{itemize}

\subsection{Variances}\label{PermutationVarianceSection}

There are $E^2$ ordered pairs of edges, and the algorithm gives contributions for each possible geometric configuration of the edges. Some algebraic simplification of the formulas below is possible; the terms have been kept separate so as to align with the cases described in the derivation in Appendix~\ref{DerivationsPermutation}

The first set of computations is performed separately at each vertex, and is derived from considering pairs of edges that meet at the vertex.
\par \medskip
{\bf Computations at a single vertex $v$}. 

\begin{align*}
\Eterm{P}_{TT}(v) ={} &\indegreev\dfrac{mn}{N(N-1)} &  \Case {1}  \\
                   &+ \indegreev  (\indegreev - 1)\dfrac{mn(m-1)}{N(N-1)(N-2)}  & \Case {2}\\
                    &+\outdegreev (\outdegreev - 1)\dfrac{mn(n-1)}{N(N-1)(N-2)}.  & \Case {3}\\
                    \\
\Eterm{P}_{CC}(v) ={} &\indegreev\dfrac{mn}{N(N-1)} & * \text{Case 1} * \\
                    &+\indegreev (\indegreev - 1)\dfrac{mn(n-1)}{N(N-1)(N-2)} & \Case {2} \\
                    &+\outdegreev(\outdegreev - 1)\dfrac{mn(m-1)}{N(N-1)(N-2)}. & \Case {3}\\
                    \\
\Eterm{P}_{TC}(v) ={} &\indegreev\outdegreev\dfrac{mn(n - 1)}{N(N-1)(N-2)}  & \Case {4}\\
                   &+\indegreev\outdegreev\dfrac{mn(m - 1)}{N(N-1)(N-2)}.  & \Case {5}\\
                   \\
\Eterm{P}_{CT}(v) ={} &\indegreev\outdegreev\dfrac{mn(m - 1)}{N(N-1)(N-2)}   & \Case {4}\\
                   &+\indegreev\outdegreev\dfrac{mn(n - 1)}{N(N-1)(N-2)}.  & \Case {5}
\end{align*}

 The total count of pairs considered is given by
\begin{align*} 
F_P ={} &\sum_v \big[\indegreev + \indegreev(\indegreev - 1) + \outdegreev(\outdegreev - 1) \\ &+ 2\indegreev\outdegreev\big].
\end{align*}

The final variance then is computed by
\begin{align}\label{varianceintermediatePermutation}
\ExpP&\left(\begin{bmatrix}W_T \\ W_C\end{bmatrix}\begin{bmatrix}W_T \\ W_C\end{bmatrix}^t\right) = \nonumber \sum_v \begin{bmatrix} \Eterm{P}_{TT}(v) & \Eterm{P}_{TC}(v) \\  \Eterm{P}_{CT}(v) & \Eterm{P}_{CC}(v) \end{bmatrix} +\\
&(E^2 - F_P)\dfrac{mn(m-1)(n - 1)}{N(N-1)(N-2)(N - 3)}\begin{bmatrix} 1 & 1 \\  1 & 1 \end{bmatrix}.
\end{align}

The last term in~\eqref{varianceintermediatePermutation} comes from Case 6 in the derivation in Appendix~\ref{DerivationsPermutation}. Finally
\begin{equation}
\VarP\left(\begin{bmatrix}W_T \\ W_C\end{bmatrix}\right) = \ExpP\left(\begin{bmatrix}W_T \\ W_C\end{bmatrix}\begin{bmatrix}W_T \\ W_C\end{bmatrix}^t\right)  - E^2
\left[\dfrac{mn}{N(N-1)}\right]^2\begin{bmatrix} 1 & 1 \\  1 & 1 \end{bmatrix}.
\end{equation}

The algorithm above, plus the discussion of complexity in section \ref{variancecomplexityallthree} completes the proof of theorem~\ref{permutationtheorem}.
\end{proof}

\section{Graphical approach to the bootstrap} 

Since the bootstrap uses the same patients as in the observed data set, the graph $\mathbb G$ described in section \ref{winsandlossessection} remains useful. 

The edges and vertices remain the same as previously. 
\begin{itemize}
\item The trial arm assignments do not change from the observed data, and we refer to {\em treatment vertices}  and {\em control vertices} based on the trial arms.  Now for each bootstrap sample there is one new piece of information: the {\em weight} of a vertex will be the number of times that the corresponding patient appears in the bootstrap sample.  

\item An edge corresponding to a treatment win or a control win in the observed data will be called a {\em treatment edge} or a {\em control edge}. 

\item There are ($\originalgraphW_T + \originalgraphW_C)^2$ ordered pairs of edges corresponding to wins. All such pairs contribute to the variance, and will be counted below. 
\begin{itemize}
\item Edges joining a pair of treatment vertices or edges joining a pair of control vertices cannot correspond to wins in any bootstrap sample, and hence do not contribute to the final variance terms. In other words, the variance of wins in a bootstrap sample does not depend on within arm comparisons.
\item This situation is in contrast to the permutation variance, where within arm comparisons explicitly contribute to the final variance.
\end{itemize}

\item For each vertex $v$, let $\adjacentTedges_v$ denote the number of treatment edges adjacent to that vertex, and let $\adjacentCedges_v$ denote the number of control edges adjacent to that vertex. 

As previously, the vector $\vectorT$ contains the number of treatment wins for each edge.  In section~\ref{winsandlossessection} a treatment edge would have had $T_e = 1$, since each patient appears once in the observed data. Now we have $T_e$ equal to the product of the vertex weights, because each instance of patient $i$ in the bootstrap sample will be preferred to each instance of patient $j$. If either patient $i$ or patient $j$ does not appear in the bootstrap, then the corresponding weight will be 0, and there are no corresponding wins; thus $T_e = 0$. The vector $\vectorC$ of control wins is extended to the bootstrap context in the same manner.  
  
\end{itemize}

 \section{Bootstrapping by sampling from the trial arms separately}\label{AnalysisBootstraptwo}

 In this section the sampling is separate in the two trial arms.   With a population of size $m$ treatment patients and $n$ control patients, there are $m^m n^n$ possible bootstrap samples, and our underlying assumption is that all are equally likely. 
 
 One could also consider bootstrap sampling from the entire population; the formulas are similar in spirit to those given here, and the details are available from the authors.

\begin{theorem} \label{bootstraptwotheorem} The mean and variance of wins under bootstrap sampling from the trial arms separately can be computed in $O(N^2)$ time.
\end{theorem}
\begin{proof} \

\subsection{Expectations}\label{Bootstraptwoexpectations}
As a first step we compute the expected values of the vectors $T$ and $C$.  
\begin{itemize}
\item For a single edge $e$ that corresponds to a treatment win, $\ExpBtwo(T_e) =  1$, and  $\ExpBtwo(C_e) =  0$.
\item  For the entire graph $\GraphG$, the expected number of treatment wins in a bootstrap sample is $\originalgraphW_T $.
\item For a single edge $e$ that corresponds to a control win, $\ExpBtwo(C_e) =  1$, and  $\ExpBtwo(T_e) =  0$.
\item  For the entire graph $\GraphG$, the expected number of control wins in a bootstrap sample is $\originalgraphW_C $.
\end{itemize}

\subsection{Variances}\label{BootstraptwoVariances}

There are $(\originalgraphW_T + \originalgraphW_C)^2$ ordered pairs of edges to consider, and the algorithm gives contributions for each possible geometric configuration of the edges. Some algebraic simplification of the formulas below is possible; the terms have been kept separate so as to align with the cases described in the derivation of Appendix~\ref{DerivationsBootstraptwo}.

\par \medskip

The vertex computations are different for treatment and control vertices.
\begin{itemize}
\item {\bf Computations at a single treatment vertex} $v$. 
\end{itemize}

\begin{align*}
\Eterm{B2}_{TT}(v) ={} &\dfrac{\adjacentTedges_v}{2} \dfrac{(2m -1)(2n - 1)}{mn}  & \Case{1}\\
 &+\adjacentTedges _v(\adjacentTedges_v - 1)\dfrac{ (2m - 1)(n - 1)}{mn}. &\Case{3} \\
 \\
 \Eterm{B2}_{CC}(v) ={} &\dfrac{\adjacentCedges_v}{2}\dfrac{(2m -1)(2n - 1)}{mn}  & \Case{2}\\
 &+\adjacentCedges _v(\adjacentCedges_v - 1)\dfrac{(2m - 1)(n - 1)}{mn}. &\Case{5} \\
 \\
\Eterm{B2}_{TC}(v) ={}& \adjacentTedges_v\adjacentCedges_v\dfrac{(2m - 1)(n - 1)}{mn}. &\Case{7}   \\
 \\
 \Eterm{B2}_{CT}(v) ={}& \adjacentCedges_v\adjacentTedges_v\dfrac{(2m - 1)(n - 1)}{mn}. &\Case{9}  
\end{align*}

\begin{itemize}
\item{\bf Computations at a single control vertex} $v$. 
\end{itemize}

\begin{align*}
\Eterm{B2}_{TT}(v) ={} &\dfrac{\adjacentTedges_v}{2} \dfrac{(2m -1)(2n - 1)}{mn}  & \Case{1}\\
 &+\adjacentTedges _v(\adjacentTedges_v - 1)\dfrac{ (m - 1)(2n - 1)}{mn}. &\Case{4} \\
 \\
 \Eterm{B2}_{CC}(v) ={} &\dfrac{\adjacentCedges_v}{2}\dfrac{(2m -1)(2n - 1)}{mn}  & \Case{2}\\
 &+\adjacentCedges _v(\adjacentCedges_v - 1)\dfrac{(m - 1)(2n - 1)}{mn}. &\Case{6} \\
 \\
\Eterm{B2}_{TC}(v) ={}& \adjacentTedges_v\adjacentCedges_v\dfrac{(m - 1)(2n - 1)}{mn}. &\Case{8}   \\
 \\
 \Eterm{B2}_{CT}(v) ={}& \adjacentCedges_v\adjacentTedges_v\dfrac{(m - 1)(2n - 1)}{mn}. &\Case{10}    
\end{align*}

{\bf Computations for the remaining edges}.
To complete the calculation we need to include pairs of non-intersecting edges.  The counts are
\begin{align*}
\remterm{B2}_{TT} &=  \originalgraphW_T(\originalgraphW_T - 1) - \sum_v \#T_v (\#T_v - 1). &\Case{11} \\
\remterm{B2}_{TC} &=  \originalgraphW_T\originalgraphW_C - \sum_v \#T_v \#C_v. &\Case{12}\\
\remterm{B2}_{CT} &=  \originalgraphW_C\originalgraphW_T - \sum_v \#C_v \#T_v. &\Case{13}\\
\remterm{B2}_{CC} &= \originalgraphW_C(\originalgraphW_C - 1) - \sum_v \#C_v (\#C_v - 1). &\Case{14}
\end{align*}

The final variance then is computed by
\begin{align*}
\ExpBone&\left(\begin{bmatrix}W_T \\ W_C\end{bmatrix}\begin{bmatrix}W_T \\ W_C\end{bmatrix}^t\right) = \sum_v \begin{bmatrix} \Eterm{B2}_{TT}(v) & \Eterm{B2}_{TC}(v) \\  \Eterm{B2}_{CT}(v) & \Eterm{B2}_{CC}(v) \end{bmatrix} +\\
&\dfrac{(m - 1)(n - 1)}{mn}\begin{bmatrix} \remterm{B2}_{TT} &  \remterm{B2}_{TC} \\ \remterm{B2}_{CT}  &\remterm{B2}_{CC} \end{bmatrix},
\end{align*}

and
\begin{equation}\label{finalvariancebootstraptwo}
\VarBtwo\left(\begin{bmatrix}W_T \\ W_C\end{bmatrix}\right) = \ExpBtwo\left(\begin{bmatrix}W_T \\ W_C\end{bmatrix}\begin{bmatrix}W_T \\ W_C\end{bmatrix}^t\right)  - 
\begin{bmatrix} \originalgraphW_T^2 & \originalgraphW_T\originalgraphW_C \\   \originalgraphW_C\originalgraphW_T & \originalgraphW_C^2 \end{bmatrix}.
\end{equation}

The algorithm above, plus the discussion of complexity in section \ref{variancecomplexityallthree} completes the proof of theorem~\ref{bootstraptwotheorem}.
\end{proof}

R code for the algorithm is given in appendix~\ref{RCodebootstraptwo}.

\section{Complexity} \label{variancecomplexityallthree}

\begin{itemize}
\item Since the mechanism of producing the matrix $U$ is undefined, all above diagonal entries must be examined at least once.  Accordingly no algorithm for the permutation or bootstrap analysis can have time complexity faster than $O(N^2)$.
\item Computing the total number of edges in the graph, and the numbers of treatment and control wins is  $O(N^2)$.
\item For each vertex, computing the indegree, the outdegree, and the counts $\#C_v$ and $\#T_v$, is performed by examining the entries in the corresponding row of the matrix $U$ and the entries in the trial arm vector. These operations are $O(N)$ for each vertex; since there are $N$ vertices the overall complexity of these computations is $O(N^2)$. 
\item The remaining computations are $O(1)$.
\item Accordingly the overall complexity of each algorithm is $O(N^2)$. 
\end{itemize}

When the Finkelstein-Schoenfeld method is used to compute the $U$ matrix, a bounded number of computations is performed for each pair of patients, so that the complexity of computing $U$ is also $O(N^2)$. One would anticipate that a considerably larger constant factor occurs in the complexity for computing $U$ than for computing the variance. Accordingly computing the variance and covariance of the wins, using any of the algorithms, would add only a negligble time to the overall computation.  

One could in principle consider all possible permutations or bootstrap samples, compute the numbers of treatment and control wins for each, and then compute means and variances. Such a computation would be exponential in $N$, and hence completely unsatisfactory for a real example. We have performed such computations in very small examples as a test of our algorithm. 

\section{Win Ratio} \label{ratiovariancesection}
Various authors including Pocock et al. \cite{PocockWinRatio},  Dong \cite{DongWinRatioVariance}, and Bebu~\cite{BebuUstatistics} examine the {\em win ratio}, which is $R_W = W_T/W_C$, and our methodology can yield information about the ratio also. As is standard with the analyses of frequencies \cite{AgrestiCategorical}, we analyze the log of the ratio because it is expected to be closer to normal than the ratio itself. 

\subsection{Pocock method}
Pocock suggests the following method for analyzing the win ratio. From the observed win ratio $R_W$, compute $\log(R_W)$.  Let $z$ be the standardized normal deviate obtained from the Finkelstein-Schoenfeld significance test for the null hypothesis of no difference between the trial arms. Then an approximate standard error for  $\log(R_W)$ is $s=\log(R_W)\div z$. This method is derived from the implicit assumption that the significance test of no difference between the trial arms should yield the same result whether working with win differences or win ratios.  In the case of ordinary frequencies, this assumption is very close to being true, and it is surely reasonable in the present context.  Once an approximate standard error has been computed, then confidence intervals can be computed and other inferences drawn.

As an alternative, Pocock suggests using the bootstrap, with separate sampling from the two treatment groups.  Our formulas for the bootstrap  could be used for this purpose. 

\subsection{Delta method}
Dong uses maximum likelihood and $U$-statistics theory to compute means and variances for $W_T$ and $W_C$, and then uses delta method transformations to compute means and variances of the ratio. Even though the present means and variances for $\Exp(W_T)$ and $\Exp(W_C)$ are based on counting rather than on maximum likelihood, the delta method transformations could still be applied to our variance formulas. The transformations indeed do produce reasonable values; we will deal with such issues in a subsequent paper.

\section{Discussion}

The methodology above gives an efficient method of computing the exact variance of the trial arm wins $W_T$ and $W_C$.  Whether or not this method constitutes a {\em closed form} depends on the definition of the term. The same comment would apply to the formula of Finkelstein-Schoenfeld for the variance of $W_T - W_C$. However, it seems to us that equation~\eqref{FSVarmainline} is considerably simpler to apply in practice than the more classical formula involving ties, as given in Hollander and Wolfe \citeyearpar[equation (4.13)]{HollanderWolfenonparametric} . Moreover, the derivation of the Finkelstein-Schoenfeld formula is vastly simpler than that for the classical formula. 

In the case of a single variable Mann-Whitney test, the variances from the classical formula, from the Finkelstein-Schoenfeld formula, and from our formula will agree, because all are based on counting arguments using the exact permutation distribution of the trial arms. An algebraic derivation of the identity would seem to be very complicated, and hardly worth the effort. One can, of course, compute as many numerical examples as desired. For this purpose the SAS PROC NPAR1WAY and R function {\tt wilcox.test} would be useful, but one must be careful to turn off the continuity correction. A similar note applies to the censored situation with a single variable. Both the Finkelstein-Schoenfeld formula and our formula implement the Gehan-Wilcoxon test as formulated by Gehan~\citeyearpar{Gehancensored}; the implementations in the SAS PROC LIFETEST and R function {\tt survdiff} are very slightly different. 

The analysis presented here using the bootstrap distribution is similar in spirit to the analysis of the permutation distribution.   Because of the special structure of the analysis considered here, we are able to determine the expected mean and variance from a bootstrap sample without actually doing the randomization. These values are effectively what would be found if an infinite number of bootstrap samples were taken.  Efron has computed exact values in a different context~\citep[ section 10.3]{EfronJackknifeBootstrap}.

Our computation of the bootstrap mean and variance is as easy as that for the permutation mean and variance used for the Mann-Whitney, Wilcoxon, or Gehan-Wilcoxon analyses. Accordingly the bootstrap evaluation could be considered a practical alternative to the original permutation test based analyses.  Because of the extensive historical use of Mann-Whitney, we would be highly reluctant to recommend making such a change. All we suggest is that one should proceed carefully in cases where the classical method and the bootstrap produce different results from a significance test. 

The derivations of the variance for wins  need to consider only pairs of edges. If one wishes to further investigate the distribution of wins, there are a number of constraints involving more than two edges.  For example, in an undirected cycle of 4 edges, it is possible for all four edges to correspond to wins. (In the example of figure~\ref{smallexamplecaption} let $D_2 = 1, D_3 = 0, D_4 = 1, \text{and } D_5 = 0$. Then the cycle gives 1 treatment win and 3 control wins.) On the other hand, in an undirected cycle of 3 edges, there can be at most 2 wins, since at least one edge must have both vertices with the same trial arm assignment.  Odd cycles play a pivotal role in maximum matching algorithms  \citep{MaximumMatchingSlither}, \citep{PathsTreesFlowers}, and conceivably they will play a role in further investigations of the distribution. 

The theoretical background for assuming normality in these analyses comes from the area of U-statistics. These items are discussed in Dong~\citeyearpar{DongWinRatioVariance}, Finkelstein-Schoenfeld~\citeyearpar{FinkelsteinSchoenfeldCombining},  Lehmann~\citeyearpar{LehmannNonparametrics}, and references therein. In this manuscript we deal with rather simpler combinatorial issues; we have nothing to add to the U-statistics literature.

\section{A small example} \label{smallexamplemain}
The small example of figure~\ref{smallexamplecaption} illustrates many of the cases above.  Because of the small size, it is possible to work out the complete permutation distribution of trial arms, all the bootstrap samples, and compare the results with the formulas.  Here we present just the results; further details are contained in the appendices.

The comparison matrix is defined by
\begin{equation*} \label{smallexampleU}
U = \begin{bmatrix}  
     \hphantom{-}0    &                  -1    &\hphantom{-}0   &\hphantom{-}0    &\hphantom{-}1  \\ 
     \hphantom{-}1     &\hphantom{-}0   &\hphantom{-}1   &\hphantom{-}0     &                  -1 \\   
     \hphantom{-}0     &                   -1   &\hphantom{-}0   &\hphantom{-}1     &\hphantom{-}0  \\
     \hphantom{-}0     &\hphantom{-}0  &                    -1    &\hphantom{-}0    &                  -1  \\ 
                       -1    &\hphantom{-}1     &\hphantom{-}0    &\hphantom{-}1    &\hphantom{-}0  
\end{bmatrix}.
\end{equation*}

The trial arms are
\begin{equation*} \label{smallexamplearms}
 D = \begin{bmatrix}  
     1 & 1 & 0 &0 &0
\end{bmatrix}.
\end{equation*}

\subsection{Permutation Distribution}
The expectation and variance are
\begin{equation}
\ExpP \begin{bmatrix} W_T \\ W_C \end{bmatrix} = \dfrac{1}{5} \begin{bmatrix} 9 \\ 9 \end{bmatrix}. 
\end{equation}

and 
\begin{equation}
\VarP \begin{bmatrix} W_T \\ W_C \end{bmatrix} = \dfrac{1}{100} \begin{bmatrix}76  &-24 \\ -24 & 56\end{bmatrix} .
\end{equation}

Then for the win difference, 
\begin{equation*}
\ExpP \left( W_T - W_C\right) = 0.
\end{equation*}

\begin{align*}
\VarP \left( W_T - W_C\right) &=  \nonumber \begin{bmatrix} 1 \\ -1 \end{bmatrix} \VarP \begin{bmatrix} W_T \\ W_C \end{bmatrix} \begin{bmatrix} 1 \\ -1 \end{bmatrix}^t \\
&= \nonumber  \dfrac{1}{100}\begin{bmatrix} 1 \\ -1 \end{bmatrix} \begin{bmatrix}76  &-24 \\ -24 & 56\end{bmatrix} \begin{bmatrix} 1 \\ -1 \end{bmatrix}^t \\
&=\dfrac{180}{100} = 1.80.
\end{align*}

The same results can be obtained by applying the Finkelstein-Schoenfeld formulas to the example.

\subsection{Bootstrap distribution -- two sample}

The expectation and variance are
 
\begin{equation}
\ExpBtwo \begin{bmatrix} W_T \\ W_C \end{bmatrix} =  \begin{bmatrix} 2 \\ 1 \end{bmatrix}. 
\end{equation}

and 
\begin{equation}
\VarBtwo \begin{bmatrix} W_T \\ W_C \end{bmatrix} = \dfrac{1}{6} \begin{bmatrix}10  &-1 \\-1 & 9\end{bmatrix}. 
\end{equation}

Then for the win difference, 
\begin{equation*}
\ExpBtwo \left( W_T - W_C\right) = 1.
\end{equation*}

\begin{align*}
\VarBtwo \left( W_T - W_C\right) &=  \nonumber \begin{bmatrix} 1 \\ -1 \end{bmatrix} \VarBtwo \begin{bmatrix} W_T \\ W_C \end{bmatrix} \begin{bmatrix} 1 \\ -1 \end{bmatrix}^t \\
&= \nonumber  \dfrac{1}{6}\begin{bmatrix} 1 \\ -1 \end{bmatrix} \begin{bmatrix}10  &-1 \\ -1 & 9\end{bmatrix} \begin{bmatrix} 1 \\ -1 \end{bmatrix}^t \\
&=\dfrac{21}{6} = 3.5.
\end{align*}

\section*{Appendices} 
There are number of appendices, which support and illustrate the derivations above.

\begin{itemize}
\item Appendix \ref{RCode}: R code for the two situations. 
\item Appendix \ref{SASCode}: SAS Code for the two situations.
\item Appendix \ref{DerivationsPermutation}: Details of wins analysis using the Permutation model.
\item Appendix \ref{DerivationsBootstraptwo}: Details of wins analysis using the two-sample bootstrap model.
\end{itemize}

In addition to these appendices, the following items are available from either of the authors.  
\begin{itemize}
\item A gentle derivation of The Finkelstein-Schoenfeld Test, in the spirit of this manuscript.
\item A derivation of the multinomial coefficient identities used in analyzing the bootstrap.
\item Formulas and detailed derivations of wins analysis for bootstrapping from the entire population, rather than from the trial arms separately.
\item R validation code for all three distributions. This file contains analysis of some examples using the the practical code, all permutations or bootstrap samples considered by hand, and randomized permutations or bootstrap samples.
\end{itemize}

\renewcommand{\thesection}{\Alph{section}}
\setcounter{section}{0}


\section{R code}\label{RCode} This appendix contains R code for the practical computations. The code uses only base R functions, and has been tested using R version 3.5.1. The code follows closely the algorithm presented in sections~\ref{AnalysisPermutation},  and \ref{AnalysisBootstraptwo}

\renewcommand\theequation{\ref{RCode}\arabic{equation}}
\setcounter{equation}{0}

\subsection{R code for the permutation distribution}\label{RCodepermutation}

\begingroup
    \fontsize{10pt}{12pt}\selectfont
\begin{verbatim}
# Compute means and variances for trial arm wins,
# based on the permutation distribution,
# using an already computed win matrix and trial arm counts. 
# The algorithm is O(N^2) in time and space.

onevertex_P <- function(indegree, outdegree){
     # indegree and outdegree correspond to a specific vertex
     # subscripts correspond to case numbering in manuscript
     c(indegree, indegree*(indegree - 1),  outdegree*(outdegree - 1), 
       indegree*outdegree, indegree*outdegree)
}

winsmeanandvariance_P <- function(winmatrix, trialarms){
     # code is modeled after manuscript section 4; some simplification is possible
     # winmatrix is the skew matrix of wins, perhaps from a hierarchical evaluation
     # m, n are the number of treatment and control patients: m + n = nrow(winmatrix)
     # function computes the expected number of treatment and control wins, and their variance
     # the actual trial assignments are not relevant for this function, just the counts
     m <- sum(trialarms == 1); n <- sum(trialarms == 0) 
     N <- m + n  # had better be the dimension of winmatrix
     indegrees <- colSums(winmatrix == 1)
     outdegrees <- rowSums(winmatrix == 1)
     # add the cases from all the vertices
     casecounts <- rowSums(mapply(onevertex_P, indegrees, outdegrees))
     expedge <- m*n/(N*(N-1)) # expected wins for one edge
     expected = casecounts[1]*rep(expedge, 2)
     names(expected) <- c("Test Wins", "Control Wins") 
     TTerms <- c(m*n/(N*(N - 1)), m*n*(m - 1)/(N*(N - 1)*(N - 2)),  
     	m*n*(n - 1)/(N*(N - 1)*(N - 2)), 0, 0)
     CTerms <- c(m*n/(N*(N - 1)), m*n*(n - 1)/(N*(N - 1)*(N - 2)),  
     	m*n*(m - 1)/(N*(N - 1)*(N - 2)), 0, 0)
     TCTerms <- c(0, 0, 0, m*n*(n - 1)/(N*(N - 1)*(N - 2)),  
     	m*n*(m - 1)/(N*(N - 1)*(N - 2)))
     CTTerms <- c(0, 0, 0, m*n*(m - 1)/(N*(N - 1)*(N - 2)),  
     	m*n*(n - 1)/(N*(N - 1)*(N - 2)))
     FP <- sum(casecounts)
     expectedforvarianceterm1 <- c(sum(casecounts*TTerms), sum(casecounts*TCTerms), 
                          sum(casecounts*TCTerms), sum(casecounts*CTerms))
     expectedforvarianceterm2 <- (casecounts[1]^2 - FP)*m*n*(m - 1)*(n - 1)/
     	(N*(N - 1)*(N - 2)*(N - 3))
     variance <- expectedforvarianceterm1 + expectedforvarianceterm2 - 
     	casecounts[1]^2*(m*n/(N*(N - 1)))^2
     variance <- matrix(variance, nrow = 2)
     rownames(variance) <- c("Test Wins", "Control Wins"); 
     colnames(variance) <- c("Test Wins", "Control Wins")
     expdiff <- expected[1] - expected[2]; names(expdiff) <- NULL
     vardiff <- variance[1, 1] - 2*variance[1, 2] + variance[2, 2]
     list(Expected = expected, Variance = variance, 
          `Expected Difference` = expdiff, `Variance Difference` = vardiff)
}

\end{verbatim}
\endgroup

\subsection {R code for the two-sample bootstrap distribution}\label{RCodebootstraptwo} 

\begingroup
    \fontsize{10pt}{12pt}\selectfont
\begin{verbatim}

onevertex_B2 <- function(T, C, arm){
     # T and C are the number of treatment and control wins involving that vertex
     # arm is the trial arm for that vertex
     # subscripts correspond to case numbering in manuscript
     if (arm == 1) cases <- c(T/2, C/2, T*(T-1), 0, C*(C-1), 0, T*C, 0, C*T, 0)
     else cases <-          c(T/2, C/2, 0, T*(T-1), 0, C*(C-1), 0, T*C, 0, C*T)
     cases
}


winsmeanandvariance_B2 <- function(winmatrix, trialarms){
     Observed <- wincomputations(winmatrix, trialarms)
     names(Observed) <- c("Test Wins", "Control Wins")
     N <- length(trialarms)
     treatmentarms <- trialarms == 1; m <- sum(treatmentarms) 
     controlarms <- trialarms == 0; n <- sum(controlarms)  
     # compute vectors giving the numbers of treatment and control wins for each vertex
     treatmentwinmatrix <- outer(treatmentarms, controlarms)*(winmatrix == 1) 
     treatmentwinmatrix <- treatmentwinmatrix + t(treatmentwinmatrix)
     controlwinmatrix <- outer(controlarms, treatmentarms)*(winmatrix == 1) 
     controlwinmatrix <- controlwinmatrix + t(controlwinmatrix)
     vertextreatmentwins <- rowSums(treatmentwinmatrix) 
     vertexcontrolwins <- rowSums(controlwinmatrix) 
     # compute treatment and control wins
     Twins <- sum(vertextreatmentwins)/2 # remove double count
     Cwins <- sum(vertexcontrolwins)/2 # remove double count
     # use the counted wins to count various cases corresponding to that vertex
     casecounts <- rowSums(mapply(onevertex_B2, vertextreatmentwins,
                                  vertexcontrolwins, trialarms)) 
     expedge <- 1 # expected wins for one edge
     expected = expedge*c(Twins, Cwins) 
     names(expected) <- c("Test Wins", "Control Wins")
     TTTerms <- c((2*m - 1)*(2*n - 1), 0, (2*m - 1)*(n - 1), (m - 1)*(2*n - 1) , 
     	rep(0,  6))/(m*n)
     CCTerms <- c(0, (2*m - 1)*(2*n - 1), 0, 0, (2*m - 1)*(n - 1),  (m - 1)*(2*n - 1), 
     	rep(0, 4))/(m*n)
     TCTerms <- c(rep(0, 6), (2*m - 1)*(n - 1), (m - 1)*(2*n - 1), rep(0, 2))/(m*n)
     CTTerms <- c(rep(0, 8), (2*m - 1)*(n - 1), (m - 1)*(2*n - 1))/(m*n)
     expectedforvarianceterm1 <- c(sum(casecounts*TTTerms), sum(casecounts*TCTerms), 
                                     sum(casecounts*TCTerms), sum(casecounts*CCTerms))
     expectedforvarianceterm2 <- (m - 1)*(n - 1)/(m*n)*
          c(Twins*(Twins - 1) - sum(casecounts[c(3, 4)]), 
                   Twins*Cwins - sum(casecounts[c(7, 8)]),
                   Cwins*Twins - sum(casecounts[c(9, 10)]), 
                   Cwins*(Cwins - 1) - sum(casecounts[c(5, 6)]))
     
     variance <- expectedforvarianceterm1 + expectedforvarianceterm2 -
          c(Twins*Twins, Twins*Cwins, Cwins*Twins, Cwins*Cwins)
     variance <- matrix(variance, nrow = 2)
     rownames(variance) <- c("Test Wins", "Control Wins") 
     colnames(variance) <- c("Test Wins", "Control Wins")  
     expdiff <- expected[1] - expected[2]
     vardiff <- variance[1, 1] - 2*variance[1, 2] + variance[2, 2]
     list(Expected = expected, Variance = variance,
          `Expected Difference` = expdiff, `Variance Difference` = vardiff)
}

\end{verbatim}
\endgroup


\section {SAS code}\label{SASCode} This section contains various items of SAS code, which have been tested using SAS version 9.4.

\renewcommand\theequation{\ref{SASCode}\arabic{equation}}
\setcounter{equation}{0}

\subsection{SAS Code for the permutation distribution}\label{SASCodepermutation}
\begin{verbatim}
/*This code creates an example U matrix*/
Data U; 
   input U1 U2 U3 U4 U5; 
   datalines;            
0 -1 0 0 1
1 0 1 0 -1
0 -1 0 1 0
0 0 -1 0 -1
-1 1 0 1 0
;                             

/*Assign values to be used later*/
%let nT=2; %let nC=3; %let N=%eval(&nT+&nC); 
%let expedge=%sysevalf(&nT*&nC/(&N*(&N-1))); 

/*Calculate the indegree, outdegree and counts for ID, 2, 3, 4 and 5*/
Data example; 
set U; 
array U U1-U&N; 
indegree=0; 
outdegree=0; 
do i=1 to &N;  
	if U[i]=-1 then indegree= indegree+1;
    if U[i]=1 then outdegree = outdegree+1;  
end; 
drop i;
case2=indegree*(indegree-1); 
case3=outdegree*(outdegree-1);
case4_5=indegree*outdegree; 
run; 

/*Sum indegree, and cases 2, 3, 4 and 5 */
proc means data=example (keep=outdegree indegree case2 case3 case4_5);
  output out=variance sum=;
run;

/*Calculate Fp counts and the variance for the number of wins for the
treatment arm (TT), the variance for the number of wins for the control arm
(CC) and their covariance (CT). Finally calculate the expectation (Exp_WD) and 
variance for the win difference (Var_WD)*/
data variance; 
set variance; 
Fp=indegree-case2-case3-2*case4_5;
factor1= &nT*(&nT - 1)*&nC/(&N*(&N - 1)*(&N - 2)); 
factor2= &nT*&nC*(&nC - 1)/(&N*(&N - 1)*(&N - 2));
TT= &expedge*indegree + factor1*case2 + factor2*case3;
CC= &expedge*indegree + factor2*case2 + factor1*case3;
TC= factor1*case4_5 + factor2*case4_5;
Exp_WD=outdegree*&expedge - indegree*&expedge;
Var_WD=TT+CC-2*TC;
drop _TYPE_ _FREQ_; 
run;

proc print data=variance (keep=Exp_WD Var_WD TT CC TC); run;
\end{verbatim}

\subsection{SAS Code for the two-sample bootstrap distribution}\label{SASCodetwosamplebootstrap}

\begin{verbatim}
/*Assign values to be used later*/
%let nT=2; %let nC=3; %let N=%eval(&nT+&nC); %let expedge=1; 
%let C=%eval(&nT+1);

/*create treatment win matrix (UT) and control win matrix (UC) from U */
Data UT1; 
set U (obs=&nT); 
array Uc U&C-U&N; 
array Ut U1-U&nT;
do i=1 to &nC; 
	if Uc[i]<0 then Uc[i]=0; 
end; 
do j=1 to &nT; 
	 Ut[j]=0; 
end;
drop i j; 
run; 

Data UT2; 
set U (firstobs=&C); 
array Uc U&C-U&N; 
array Ut U1-U&nT;
do i=1 to &nC; 
	Uc[i]=0;
end; 
do j=1 to &nT; 
	if Ut[j]<0 then Ut[j]=1; else Ut[j]=0;  
end;
drop i j; 
run; 

Data UT;
set UT1 UT2; 
run;

Data UC1; 
set U (obs=&nT); 
array Uc U&C-U&N; 
array Ut U1-U&nT;
do i=1 to &nC; 
	if Uc[i]=-1 then Uc[i]=1; else Uc[i]=0; 
end; 
do j=1 to &nT; 
	 Ut[j]=0; 
end;
drop i j; 
run; 

Data UC2; 
set U (firstobs=&C); 
array Uc U&C-U&N; 
array Ut U1-U&nT;
do i=1 to &nC; 
	Uc[i]=0;
end; 
do j=1 to &nT; 
	if Ut[j]<0 then Ut[j]=0;  
end;
drop i j; 
run; 

Data UC;
set UC1 UC2; 
run;

/*sum the wins per row for both win matrices*/
Data ST; 
set UT; 
array U U1-U5; 
T=0;  
do i=1 to &N;  
	if U[i]=1 then T= T+1;  
end; 
drop i;
run; 

Data SC; 
set UC; 
array U U1-U5; 
C=0;  
do i=1 to &N;  
	if U[i]=1 then C= C+1;  
end; 
drop i;
run; 

/*merge the two sums columns and calculate case 3-10 */
data S; 
merge ST(keep=T) SC(keep=C); 
run; 
data sumT; 
set S (obs=&nT); 
case3=T*(T-1);
case5=C*(C-1);
case7_9=T*C;
run;

data sumC; 
set S (firstobs=&C); 
case4=T*(T-1);
case6=C*(C-1);
case8_10=T*C;
run;
 
data cases; 
set sumT sumC; 
run;

proc means data=cases; 
output out=variance sum=; 
run; 

/*Calculate case 1-10 counts and the variance for the number of wins for the
treatment arm (TT), the variance for the number of wins for the control arm
(CC) and their covariance (CT or TC). Finally calculate the expectation (Exp_WD) 
and variance for the win difference (Var_WD)*/
data variance; 
set variance;  
case1=T/2;
case2=C/2; 
factor0 = (2*&nT-1)*(2*&nC - 1)/(&nT*&nC);
factor1t = (2*&nT-1)*(&nC-1)/(&nT*&nC);
factor1c = (&nT-1)*(2*&nC-1)/(&nT*&nC);
factor2 = (&nT-1)*(&nC-1)/(&nT*&nC);
remTT=case1*(case1-1)-(case3+case4);
remCC=case2*(case2-1)-(case5+case6);
remTC=case1*case2-(case7_9+case8_10); 
TT = case1*factor0 + case3*factor1t + case4*factor1c + remTT*factor2-case1**2;
CC = case2*factor0 + case5*factor1t + case6*factor1c + remCC*factor2-case2**2;
TC = case7_9*factor1t + case8_10*factor1c + remTC*factor2-case1*case2;
Exp_WD=case1*&expedge - case2*&expedge;
Var_WD=TT+CC-2*TC;
drop _TYPE_ _FREQ_; 
run;

proc print data=variance (keep=Exp_WD Var_WD); run;

\end{verbatim}

\section{Detailed wins analysis using the permutation model} \label{DerivationsPermutation}
\renewcommand\theequation{\ref{DerivationsPermutation}\arabic{equation}}
\setcounter{equation}{0}

In this section we give detailed derivations of the formulas used in section \ref{AnalysisPermutation}.

\subsection{Expectations}

The total number of trial arm assignments is $\binom{N}{m}$. For a specific edge  $e = (i, j)$  the number of assignments with $D_i = 1$ and $D_j = 0$ is 
$\binom{N - 2}{m - 1}$.   Hence for all edges $e$, 
\begin{equation}\label{ExpTeP}
\ExpP(T_e) =  \binom{N-2}{m-1}\bigg/\binom{N}{m} = \dfrac{mn}{N(N-1)},  
\end{equation}
 and 
 
\begin{equation} \label{ExpW_T}
\ExpP(W_T) = E\dfrac{mn}{N(N-1)}.
\end{equation}

Since having  $D_i = 1$ and $D_j = 0$ is equally likely to having  $D_i = 0$ and $D_j = 1$, it follows that $\ExpP(C_e) = \ExpP(T_e)$, and $\ExpP(W_C) = \ExpP(W_T)$.  

\subsection{Variances}\label{PermutationVarianceSectionDetails}

The variance is found by considering the geometric configuration of pairs of edges.  The description for each case is merely for ease of following the computations; the actual algorithm uses only the counts and the expected value formulas.  

The first computations are performed separately at each vertex, and the formulas below give $\ExpP(T_eT_f)$, where $(e, f)$ is an ordered pair of edges, possibly the same.

\par \medskip
{\bf Computations at a single vertex $v$}. 

\begin{enumerate}
\renewcommand{\labelenumi}{Case \arabic{enumi}:}
\item A single edge $e$ with head $v$:
\begin{itemize}
\item The term appears $\indegreev$times.  

\item Since $T_e$ is 0 or 1, we have using~\eqref{ExpTeP}
\begin{equation*} 
\ExpP(T_e^2) = \ExpP(T_e) = \dfrac{mn}{N(N-1)}. 
\end{equation*}

\item Similarly, 
\begin{equation*}
\ExpP(C_e^2) = \ExpP(C_e) = \dfrac{mn}{N(N-1)}. 
\end{equation*}

\item Since an edge cannot be simultaneously a win for treatment and for control, we have $T_e C_e = 0$ for all edges $e$ and all permutations, and thus
\begin{equation*}
\ExpP(T_e C_e) =  \ExpP(C_e T_e) =0.                       
\end{equation*}
\end{itemize}

\item An ordered pair $(e, f)$ of distinct edges, each with head $v$:
\begin{itemize}
\item Edges $(3, 4)$ and $(5, 4)$ in figure \ref{smallexamplecaption} are an example of distinct edges with a common head. 
\item  Because the pair is ordered, both $(e, f)$ and $(f, e)$  will appear in the computation. Accordingly the term appears  $\indegreev(\indegreev - 1)$ times.
\item  Consider an ordered pair of edges $e = (v_i, v_j)$ and $f = (v_k, v_j)$ where $\{v_i, v_j, v_k\}$ are distinct; that is vertex $v_j$ is the head for both edges.  Then $T_e T_f = 1$ iff $D_i = D_k = 1$ and $D_j = 0$.  There remain $N - 3$ patients to be assigned trial arms, with $m - 2$ treatment patients and $n - 1$ control patients. The number of trial arm permutations satisfying these conditions is $\binom{N - 3}{m - 2}$.  Accordingly.

\begin{equation*}
\ExpP(T_e T_f) = \binom{N - 3}{m - 2} \bigg/ \binom{N}{m} = \dfrac{m n(m - 1)}{N (N - 1) (N - 2)}.
\end{equation*}
\item Similarly,
\begin{equation*}
\ExpP(C_e C_f) = \binom{N - 3}{n - 2} \bigg/ \binom{N}{m} =  \dfrac{m n (n - 1)}{N (N - 1) (N - 2)}.
\end{equation*}

\item $T_e$ and $C_f$ cannot both be 1, because the former requires $D_j = 0$, and the latter requires $D_j = 1$. Hence $T_e C_f = 0$ always.   Accordingly 
\begin{equation*}
\ExpP(T_e C_f) = \ExpP(C_e T_f) =0.
\end{equation*}
\end{itemize}

\item An ordered pair $(e, f)$ of distinct edges, each with tail $v$:
\begin{itemize}
\item Edges $(5, 2)$ and $(5, 4)$ in figure \ref{smallexamplecaption} are an example of distinct edges with a common tail. 
\item  Because the pair is ordered, both $(e, f)$ and $(f, e)$  will appear in the computation. Accordingly the term appears  $\outdegreev (\outdegreev - 1)$ times in the computations for vertex $v$.
\item  Consider an ordered pair of edges $e = (v_i, v_j)$ and $f = (v_i, v_l)$ where $\{v_i, v_j, v_l\}$ are distinct; that is  vertex $v_i$ is the tail for both edges.  Then $T_e T_f = 1$ iff $D_i = 1$ and $D_j = D_l = 0$.  The number of trial arm permutations satisfying these conditions is $\binom{N - 3}{m - 1}$.  Accordingly
\begin{equation*}
\ExpP(T_e T_f) = \binom{N - 3}{m - 1} \bigg/ \binom{N}{m} =  \dfrac{m  n (n - 1)}{N (N - 1) (N - 2)}.
\end{equation*}

Similarly, 
\begin{equation*} 
\ExpP(C_e C_f) = \binom{N - 3}{n - 1} \bigg/ \binom{N}{m} = \dfrac{m n (m - 1)  }{N (N - 1) (N - 2)}.
\end{equation*}

\item  $T_e$ and $C_f$ cannot both be 1, because the former requires $D_i = 1$, and the latter requires $D_i = 0$. Hence $T_e C_f = 0$ always.   Accordingly 
\begin{equation*}
\ExpP(T_e C_f) = \ExpP(C_e T_f) =0.
\end{equation*}
\end{itemize}
 
\item An ordered pair $(e, f)$ of distinct edges, where vertex $v$ is the tail of edge $e$ and the head of edge $f$:
\begin{itemize}
\item Edges $(1, 5)$ and $(2, 1)$ in figure \ref{smallexamplecaption} are an example of this situation. 
\item The term appears  $\indegreev\outdegreev$ times.
\item Consider a pair of edges $e = (v_i, v_j)$ and $f = (v_k, v_i)$ where $\{v_i, v_j, v_l\}$ are distinct; that is vertex $v_i$ is the tail of edge $e$ and the head of edge $f$.  Then $T_e$ and $T_f$ can never both be 1, because the former requires $D_i = 1$ and the latter requires $D_i = 0$. Accordingly

\begin{equation*} 
\ExpP(T_e T_f) = \ExpP(C_e C_f) =0.
\end{equation*} 

\item $T_e C_f = 1$ iff $D_i =  1$ and $D_j = D_k = 0$ .   Accordingly
\begin{equation*} 
\ExpP(T_e C_f)  = \binom{N - 3}{m - 1} \bigg/ \binom{N}{m} = \dfrac{m  n (n - 1)}{N (N - 1) (N - 2)}.
\end{equation*}

\item $C_e T_f = 1$ iff $D_i =  0$ and $D_j = D_k = 1$ .   Accordingly
\begin{equation*} 
\ExpP(C_e T_f)  = \binom{N - 3}{n - 1} \bigg/ \binom{N}{m} = \dfrac{m n (m - 1) }{N (N - 1) (N - 2)}.
\end{equation*}

\end{itemize}

\item An ordered pair $(e, f)$ of distinct edges, where vertex $v$ is the head of edge $e$ and the tail of edge $f$:
\begin{itemize}
\item Edges $(2, 1)$ and $(1, 5)$ in figure \ref{smallexamplecaption} are an example of this situation. Note that these are the same edges mentioned in the previous case, but the order is different.
\item The term appears  $\indegreev\outdegreev$ times.
\item Consider a pair of edges $e = (v_i, v_j)$ and $f = (v_j, v_k)$ where $\{v_i, v_j, v_ k\}$ are distinct; that is vertex $v_j$ is the head of edge $e$ and the tail of edge $f$.  Then $T_e$ and $T_f$ can never both be 1, nor can $C_e$ and $C_f$. Accordingly

\begin{equation*} 
\ExpP(T_e T_f) = \ExpP(C_e C_f) =0.
\end{equation*} 

\item $T_e C_f = 1$ iff $D_i  = D_k = 1$ and $D_j =  0$ .   Accordingly
\begin{equation*} 
\ExpP(T_e C_f)  = \binom{N - 3}{m - 2} \bigg/ \binom{N}{m} = \dfrac{m n(m - 1)}{N (N - 1) (N - 2)}.
\end{equation*}
Similarly, 
\item \begin{equation*} 
 \ExpP(C_e T_f)= \binom{N - 3}{n - 2} \bigg/ \binom{N}{m} =  \dfrac{m  n (n - 1)}{N (N - 1) (N - 2)}.
\end{equation*}
\end{itemize}
\end{enumerate}

\par \medskip
{\bf Computations for the remaining edge pairs}. 
\begin{enumerate}
\renewcommand{\labelenumi}{Case \arabic{enumi}:}
\setcounter{enumi}{5}
\item An ordered pair $(e, f)$ of non-intersecting edges:
\begin{itemize}
\item Edges $(2, 3)$ and $(1, 5)$ in figure \ref{smallexamplecaption} are an example of non-intersecting edges. 
\item The computations in cases 1 - 5 above have considered all possibilities for single edges and for intersecting edges.  There are a total of $E^2$ ordered pairs of edges, the total count of pairs considered is given by
\begin{align*} 
F ={} &\sum_v \big[\indegreev + \indegreev(\indegreev - 1) + \outdegreev(\outdegreev - 1) \\ &+ 2\indegreev\outdegreev\big].
\end{align*}
The remaining count is $E^2 - F$.   

\item Consider a pair of edges $e = (v_i, v_j)$ and $f = (v_k, v_l)$ where $\{v_i, v_j, v_k, v_l\}$ are distinct; that is the edges do not meet.  Then $T_e T_f = 1$ iff $D_i = D_k= 1$ and $D_j = D_l = 0$.  The number of trial arm permutations satisfying these conditions is $\binom{N - 4}{m - 2}$.  The same counting argument applies to $C_e C_f$, $T_e C_f$, and $C_e T_f$.  Accordingly

\begin{align*}
\ExpP(T_e T_f) &= \ExpP(T_e C_f) = \ExpP(C_e T_f) = \ExpP(C_e C_f) \\
                      &=   \binom{N - 4}{m - 2} \bigg/ \binom{N}{m} \\
                      &= \dfrac{mn(m - 1) (n - 1)}{N(N-1)(N - 2)(N - 3)} .
\end{align*}
\end{itemize} 
\end{enumerate}


\section{Detailed wins analysis using the two-sample bootstrap} \label{DerivationsBootstraptwo}
\renewcommand\theequation{\ref{DerivationsBootstraptwo}\arabic{equation}}
\setcounter{equation}{0}
In this section we give detailed derivations of the formulas used in section \ref{AnalysisBootstraptwo}.  The computations use elementary sums of multinomial coefficients; these are the same sums that one would evaluate in computing moments of the multinomial distribution.

\subsection{Expectations}\label{BootstraptwoexpectationsDetails}
Consider a treatment edge $e$. Then the head corresponds to a control patient, and the tail to a treatment patient. The sampling for head and tail is from the two different trial arms, and for the expected value of $T_e$, we have

\begin{align}\label{Bootstraptwoexpectationformula}
\ExpBtwo(T_e) &= \nonumber \dfrac{1}{m^m n^n}  \sum_{k_1 + \dots k_m = m} k_1 \binom{m}{k_1, k_2, \dots, k_m} 
                                                                              \sum_{l_1 + \dots l_l = l} l_1 \binom{n}{l_1, l_2, \dots, l_m}\\
                      &= \dfrac{1}{m^m n^n} m^m n^n = 1.
\end{align}
For this edge we will have $\ExpBtwo(C_e) = 0$, since this edge cannot correspond to a control win in any bootstrap sample. Similarly, if an edge $f$ is a control edge, then $\ExpBtwo(T_f) = 0$ and $\ExpBtwo(C_f) = 1$.

\subsection{Variances}\label{Bootstraponetwovariances}

\par \medskip
{\bf Computations at a single vertex}. 

\begin{enumerate}
\renewcommand{\labelenumi}{Case \arabic{enumi}:}

\item A treatment edge, adjacent to the vertex $v$:
\begin{itemize}
\item The term appears $\#T_v/2$ times. The factor $1/2$ is due to the fact that the same edge will be counted at its other end also.
\item Let the edge be $e = (v_1, v_2)$, with corresponding vertex weights $k_1, k_2$.   The number of possibilities for  $T_e >0$ is $k_1k_2$.   Then
\begin{align} \label{varsingleedgebootstraptwo}
\ExpBtwo(T_e^2) &= \nonumber \dfrac{1}{m^m n^n}  \sum_{k_1 + \dots k_m = m} k_1^2 \binom{m}{k_1, k_2, \dots, k_m} \\
                               &\nonumber\qquad \qquad \times\sum_{l_1 + \dots l_l = l} l_1^2 \binom{n}{l_1, l_2, \dots, l_m}\\
                      &= \nonumber \dfrac{1}{m^m n^n} (2m - 1)m^{m-1}(2n-1)n^{n-1} \\
                      &= \dfrac{(2m -1)}{m}\dfrac{(2n -1)}{n}.
\end{align}

For this same edge, $C_e = 0$ for all bootstrap samples, and hence 
\begin{equation*} 
\ExpBone(C_e^2) = \ExpBone(T_eC_e) =  \ExpBone(C_eT_e) = 0. 
\end{equation*}
\end{itemize}

\item A control edge, adjacent to the vertex $v$:
\begin{itemize}
\item The term appears $\#C_v/2$ times. The factor $1/2$ is due to the fact that the same edge will be counted at its other end also.
\item The remaining computations are similar to case 1.
\begin{equation*} 
\ExpBone(T_e^2)  =   \ExpBone(T_eC_e) =  \ExpBone(C_eT_e)  = 0.
\end{equation*}
\begin{equation*} 
\ExpBone(C_e^2)  =   \dfrac{(2m -1)}{m}\dfrac{(2n -1)}{n}.
\end{equation*}
\end{itemize}

\item \label{B2TTT} An ordered pair $(e, f)$ of distinct treatment edges, each adjacent to the treatment vertex $v$:
\begin{itemize}
\item The term appears  $\#T_v (\#T_v - 1)$ times. 
\item Let the edges be $e = (v_1, v_2)$ and $f = (v_1, v_3)$, with corresponding vertex weights are $k_1, k_2, k_3$.   Then the common vertex $v_1$ is a treatment vertex.  The number of possibilities for $T_e >0$ is $k_1k_2$, and for $T_f > 0$ is $k_1k_3$.   Then the number of samples that produce both $T_e > 0$ and $T_f > 0$ is $k_1^2k_2k_3$, and  
\begin{align*}
\ExpBtwo(T_eT_f) &= \nonumber\dfrac{1}{m^mn^n}\sum_{k_1 + \dots k_m = m} k_1^2\binom{N}{k_1, k_2, \dots, k_N} \\
&\qquad \qquad \times\sum_{l_1 + \dots l_n = n} l_2 l_3\binom{n}{l_1, l_2, \dots, l_n}\\
 &=  \dfrac{1}{m^mn^n} (2m-1)m^{m-1}(n-1)n^{n-1} \\
    &=   \dfrac{(2m-1)(n-1)}{mn}.
 \end{align*}   
Because both edges correspond to treatment wins, we have
\begin{align*}
\ExpBtwo(C_eC_f) = \ExpBtwo(T_eC_f) = \ExpBtwo(C_eT_f) = 0.  
\end{align*}
\end{itemize}

\item \label{B2TTC} An ordered pair $(e, f)$ of distinct treatment edges, each adjacent to the control vertex $v$:
\begin{itemize}
\item The term appears  $\#T_v (\#T_v - 1)$ times. 
\item Let the edges be $e = (v_1, v_2)$ and $f = (v_3, v_2)$, with corresponding vertex weights are $k_1, k_2, k_3$.   Then the common vertex $v_2$ is a control vertex.  The number of possibilities for $T_e >0$ is $k_1k_2$, and for $T_f > 0$ is $k_3k_2$.   Then the number of samples that produce both $T_e > 0$ and $T_f > 0$ is $k_1 k_2^2 k_3$, and  
\begin{align*}
\ExpBtwo(T_eT_f) &= \nonumber\dfrac{1}{m^mn^n}\sum_{k_1 + \dots k_m = m} k_1 k_3\binom{N}{k_1, k_2, \dots, k_N} \\
&\qquad \qquad \times \sum_{l_1 + \dots l_n = n} l_2 ^2\binom{n}{l_1, l_2, \dots, l_n}\\
 &= \dfrac{1}{m^mn^n} (m-1)m^{m-1}(2n-1)n^{n-1} \\
    &=   \dfrac{(m-1)(2n-1)}{mn}
\end{align*}
Because both edges correspond to treatment wins, we have 
\begin{align*}\ExpBtwo(C_eC_f) = \ExpBtwo(T_eC_f) = \ExpBtwo(C_eT_f) = 0.  
\end{align*}
\end{itemize}

\item An ordered pair $(e, f)$ of distinct control edges, each adjacent to the treatment vertex $v$:
\begin{itemize}
\item The term appears  $\#C_v (\#C_v - 1)$ times. 
\item The expected value computations are similar to those of case~\ref{B2TTT}.
\begin{align*}
\ExpBtwo(C_eC_f)  &= \dfrac{(2m-1)(n-1)}{mn}. \\
\ExpBtwo(T_eT_f) &= \ExpBtwo(T_eC_f) = \ExpBtwo(C_eT_f) = 0. 
\end{align*} 
\end{itemize}

\item An ordered pair $(e, f)$ of distinct control edges, each adjacent to the control vertex $v$:
\begin{itemize}
\item The term appears  $\#C_v (\#C_v - 1)$ times. 
\item The expected value computations are similar to those of case~\ref{B2TTC}.
\begin{align*}
\ExpBtwo(C_eC_f) &= \dfrac{(m-1)(2n-1)}{mn}.\\
\ExpBtwo(T_eT_f) &= \ExpBtwo(T_eC_f) = \ExpBtwo(C_eT_f) = 0.  
\end{align*}
\end{itemize}

\item  An ordered pair $(e, f)$ of edges, where $e$ is a treatment edge and $f$ is a control edge, each adjacent to the treatment vertex $v$:
\begin{itemize}
\item The term appears  $\adjacentTedges_v\adjacentCedges_v $ times.
\item The expected value computations are similar to those of case~\ref{B2TTT}.
\begin{align*}
\ExpBtwo(T_eC_f)  &= \dfrac{(2m-1)(n-1)}{mn}.\\
\ExpBtwo(T_eT_f) &= \ExpBtwo(C_eT_f) = \ExpBtwo(C_eC_f) = 0.
\end{align*}  
\end{itemize}

\item  An ordered pair $(e, f)$ of edges, where $e$ is a treatment edge and $f$ is a control edge, each adjacent to the control vertex $v$:
\begin{itemize}
\item The term appears  $\adjacentTedges_v\adjacentCedges_v $ times.
\item The expected value computations are similar to those of case~\ref{B2TTC}.
\begin{align*}
\ExpBtwo(T_eC_f) &= \dfrac{(m-1)(2n-1)}{mn}. \\
\ExpBtwo(T_eT_f) &= \ExpBtwo(C_eT_f) = \ExpBtwo(C_eC_f) = 0.  
\end{align*}
\end{itemize}

\item  \label{itemBtwo6t}An ordered pair $(e, f)$ of edges, where $e$ is a control edge and $f$ is a treatment edge, each adjacent to the treatment vertex $v$:
\begin{itemize}
\item The term appears  $\adjacentCedges\adjacentTedges_v $ times.
\item The expected value computations are similar to those of case~\ref{B2TTT}.
\begin{align*}
\ExpBtwo(C_eT_f) &= \dfrac{(2m-1)(n-1)}{mn}.\\
\ExpBtwo(T_eT_f) &= \ExpBtwo(T_eC_f) = \ExpBtwo(C_eC_f) = 0.  
\end{align*}
\end{itemize}

\item  \label{itemBtwo6c}An ordered pair $(e, f)$ of edges, where $e$ is a control edge and $f$ is a treatment edge, each adjacent to the control vertex $v$:
\begin{itemize}
\item The term appears  $\adjacentCedges\adjacentTedges_v $ times.
\item The expected value computations are similar to those of case~\ref{B2TTC}.
\begin{align*}
\ExpBtwo(C_eT_f) &= \dfrac{(m-1)(2n-1)}{mn}. \\
\ExpBtwo(T_eT_f) &= \ExpBtwo(T_eC_f) = \ExpBtwo(C_eC_f) = 0. 
\end{align*} 
\end{itemize}

\end{enumerate}

\par \medskip
{\bf Computations for the remaining edge pairs}. 

\begin{enumerate}
\renewcommand{\labelenumi}{Case \arabic{enumi}:}
\setcounter{enumi}{10}
\item  An ordered pair $(e, f)$ of non-intersecting edges, both treatment edges:
\begin{itemize}
\item  The term appears  $\originalgraphW_T(\originalgraphW_T - 1) - \sum_v \#T_v (\#T_v - 1)$ times.  This count comes from $\originalgraphW_T(\originalgraphW_T - 1)$, which is the total number of ordered pairs of distinct treatment edges, minus the pairs considered in cases 3 and 4 above.

\item Let the edges be $e = (v_1, v_2)$ and $f = (v_3, v_4)$, with corresponding vertex weights are $k_1, k_2, k_3, k_4$.   The number of possibilities for  $T_e >0$ is $k_1k_2$, and for $T_f > 0$ is $k_3k_4$.   Then the number of samples that produce both $T_e > 0$ and $T_f > 0$ is $k_1k_2k_3k_4$, and
\begin{align*} 
\ExpBtwo(T_eT_f) &= \nonumber\dfrac{1}{m^mn^n}\sum_{k_1 + \dots k_m = m} k_2 k_4\binom{N}{k_1, k_2, \dots, k_N} \\
&\qquad \qquad \times \sum_{l_1 + \dots l_n = n} l_1 l_3\binom{n}{l_1, l_2, \dots, l_n}\\
 &= \dfrac{1}{m^mn^n} (m-1)m^{m-1}(n-1)n^{n-1} \\
    &=   \dfrac{(m-1)(n-1)}{mn}.
\end{align*}
Because both edges correspond to treatment wins, we have 
\begin{equation*} \ExpBtwo(C_eC_f) = \ExpBtwo(T_eC_f) = \ExpBtwo(C_eT_f) = 0. \end{equation*}
\end{itemize}

\item An ordered pair $(e, f)$ of non-intersecting edges, both control edges. Reasoning as in case 11 we have

The term appears  $\originalgraphW_C(\originalgraphW_C - 1) - \sum_v \#C_v (\#C_v - 1)$ times. 
\begin{align*} 
\ExpBone(T_e T_f) = \ExpBone(T_eC_f) = \ExpBone(C_eT_f)= 0.
\end{align*}
\begin{equation*} \ExpBone(C_eC_f) = \dfrac{(m-1)(n-1)}{mn}. \end{equation*}

\item An ordered pair $(e, f)$ of non-intersecting edges, where $e$ is a treatment edge and $f$ is a control edge. Reasoning as in case 11 we have

The term appears  $\originalgraphW_T\originalgraphW_C  - \sum_v \#T_v \#C_v$ times. 

\begin{align*} 
\ExpBone(T_e T_f) = \ExpBone(C_eT_f) = \ExpBone(C_e C_f) =  0.
\end{align*}
\begin{equation*} \ExpBone(T_eC_f) =  \dfrac{(m-1)(n-1)}{mn}. \end{equation*}

\item An ordered pair $(e, f)$ of non-intersecting edges, where $e$ is a control edge and $f$ is a treatment edge. Reasoning as in case 11 we have

The term appears  $\originalgraphW_C\originalgraphW_T  - \sum_v \#C_v \#T_v$ times. 

\begin{align*} 
\ExpBone(T_e T_f) = \ExpBone(T_eC_f) = \ExpBone(C_e C_f) =  0.
\end{align*}
\begin{equation*} \ExpBone(C_eT_f) =  \dfrac{(m-1)(n-1)}{mn}. \end{equation*}

\end{enumerate}
\end{document}